\pdfoutput=1
\documentclass[10pt,pra,aps,superscriptaddress,nofootinbib,twocolumn,longbibliography]{revtex4-1}
\usepackage{amsmath,bbm}
\usepackage{amsthm}
\usepackage{amsmath}
\usepackage{latexsym}
\usepackage{amssymb}
\usepackage{float}
\usepackage{braket}
\usepackage{quantikz}
\usepackage{tikz}
\usepackage{graphicx}           
\usepackage{textcomp}
\usepackage{color}
\usepackage{xcolor}
\usepackage{mathpazo}
\usepackage{comment}
\usepackage{enumitem}
\usepackage{multirow}
\usepackage{setspace}
\usepackage{subcaption}
\usepackage[colorlinks=true,linkcolor=blue,citecolor=magenta,urlcolor=blue]{hyperref}
\usepackage{svg}

\newcommand{\be}{\begin{equation}}
\newcommand{\ee}{\end{equation}}
\newcommand{\bea}{\begin{eqnarray}}
\newcommand{\eea}{\end{eqnarray}}
\def\squareforqed{\hbox{\rlap{$\sqcap$}$\sqcup$}}
\def\qed{\ifmmode\squareforqed\else{\unskip\nobreak\hfil
\penalty50\hskip1em\null\nobreak\hfil\squareforqed
\parfillskip=0pt\finalhyphendemerits=0\endgraf}\fi}
\def\endenv{\ifmmode\;\else{\unskip\nobreak\hfil
\penalty50\hskip1em\null\nobreak\hfil\;
\parfillskip=0pt\finalhyphendemerits=0\endgraf}\fi}

\newcommand{\I}{\mathbbm{1}}

\newcommand{\la}{\langle}
\newcommand{\ra}{\rangle}

\makeatletter
\newtheorem*{rep@theorem}{\rep@title}
\newcommand{\newreptheorem}[2]{%
\newenvironment{rep#1}[1]{%
 \def\rep@title{#2 \ref{##1}}%
 \begin{rep@theorem}}%
 {\end{rep@theorem}}}
\makeatother

\newtheorem{thm}{Theorem}
\newreptheorem{thm}{Theorem}

\newtheorem*{example*}{Example}

\begin{document}

\title{Local Marking of Locally Implementable Unitary Operations}

\author{Adil Imam}
\affiliation{Indian Institute of Technology Kanpur, Kalyanpur, Kanpur 208016, India}
 \author{Satyaki Manna}
 \affiliation{Department of Physics, School of Basic Sciences, Indian Institute of Technology Bhubaneswar, Odisha 752050, India}

\begin{abstract}
We investigate the task of local marking for locally implementable unitary operations. In this setting, multipartite quantum unitary channels, chosen randomly from a known set, are distributed among spatially separated parties without revealing their identities. The objective is to correctly identify (mark) the applied process using only local operations supplemented with classical communication (LOCC). While local distinguishability implies local marking, local marking does not guarantee either local or even global distinguishability of a set of unitaries. Thus the task of marking is not equivalent to the task of discrimination. We demonstrate a stronger manifestation of \emph{nonlocality without entanglement} by constructing a set of globally distinguishable tripartite product unitaries that cannot be locally marked. In contrast to state marking, we find that marking a subset of product unitaries does not imply the ability to mark a larger subset. Finally, we explore the hierarchy of probes—entangled and product—in the context of local marking with respect to the standard discrimination scenario.
\end{abstract}
\maketitle


\section{Introduction}
Discrimination of different physical processes has  been an important avenue of research for understanding the subtlety of quantum theory and consequently, it has been proved to be useful in many information-theoretic tasks such as quantum cryptography\cite{Bae_2015,JánosABergou_2007} and quantum communication \cite{Manna,manna2026,pandit,pandit2025}. In classical theory, the notion of discrimination is straightforward. In contrast, this concept offered by quantum theory is significantly more peculiar. A well-known counterintuitive phenomenon is \emph{nonlocality without entanglement}\cite{hardy,benett,vedral,Virmani_2001,watrous05,Halder,ghosh_s,Bhattacharya,manna20,manna4,He2024,niset}, which shows that a set of distinguishable multipartite processes cannot be perfectly distinguished if spatially separated parties are restricted to local operations and classical communication (LOCC). Beyond such discrimination tasks, \emph{Sen et al.} introduced a variant called \emph{local state marking}\cite{Sen,sen2}. In this task, a subset of states, chosen randomly from a known set of orthogonal multipartite states, is given to spatially separated parties without revealing their identities, and the goal is to mark their identities using LOCC. We extend this novel marking game in the context of quantum unitary channels.

Although discrimination of quantum channels is well studied\cite{AYuKitaev_1997,acin2001,duan2009,watrous2,njp2021,Manna_2,manna2025,manna25,manna3,npj,harrow,feng,GMauro}, their distinguishability under the paradigm of LOCC \cite{Duan3,manna4,Heng} remains much less explored due to the higher complexity compared to state discrimination. For the same reason, other discrimination protocols for channels are largely absent in the literature. While the definition of marking used here applies to any set of unitaries, we focus on locally implementable unitaries for most of our results. Such locally implementable (product) unitaries make LOCC protocols more subtle by introducing an additional layer of communication, as demonstrated in \cite{manna4}.

Unitary operators represent one of the most fundamental quantum channels and have a significant importance in quantum communication \cite{densecoding,Brub}, information processing \cite{science.1090790}, error correction \cite{You_2012}, cryptography \cite{ekert}, and computation\cite{HUANG2022127863}. While there exists a body of research on the discrimination or antidiscrimination of unitary operators \cite{acin2001,manna25,manna2025,GMauro,Duan3,feng,manna4} in various contexts, the problem of marking provides a new class of protocols which can be useful in future informational tasks. 

In this work, we systematically formulate the problem of local marking of locally implementable unitary operations. After introducing the necessary preliminaries, we review two discrimination protocols—global and LOCC—already studied in the literature. We then define the marking protocol: a subset of unitaries, chosen randomly from a known set (not necessarily distinguishable), is given to spatially separated parties without revealing their identities, and the goal is to mark them using LOCC. This differs from the state-marking scenario, where the given set is assumed to be distinguishable.
In the first part of our results, we establish relations among global discrimination, LOCC discrimination, and local marking. We show that local discrimination always implies local marking. However, there exist sets of unitaries that are markable but neither locally nor globally distinguishable. We provide a necessary and sufficient condition for two globally indistinguishable unitaries to be markable. We also show that any two bipartite globally distinguishable unitaries are always markable. In contrast, for more than two parties, there exist sets of unitaries that are globally distinguishable but not markable. This result can be viewed as a parallel to the phenomenon of \emph{nonlocality without entanglement}. In our case, we introduce a task that is strictly stronger than local discrimination; nevertheless, it remains strictly weaker than global discrimination.

We further show that if a subset of unitaries is markable, then any smaller subset is also markable, while the converse does not hold. In particular, we construct a set of three unitaries such that any pair can be marked, but all three cannot be marked together. This phenomenon emerges strikingly opposite to the case of quantum states. Finally, we compare two scenarios: marking with single system (product) probes and discrimination with entangled probes. We present a set of unitaries that cannot be distinguished using product probes but can still be marked using them.

Before diving to our main motive, we discuss some known protocols regarding the discrimination of unitary operations.


\section{Preliminaries}

\subsection{Discrimination of Unitary operations}\label{ssA}
Unitary operator $(U)$ is a linear operator $U:H\rightarrow H$ on a Hilbert space $H$ that satisfies $U^\dagger U=UU^\dagger=\I$.
We consider a previously known set of $m$ unitaries $\{U_i\}_{i=1}^m$ acting on a $d$-dimensional quantum state. The unitaries are sampled from a given probability distribution. To distinguish the unitaries, a known quantum state (single or entangled) is given into the unitary device and the device carries out one of $m$ unitary operations. After this process, the device gives an transformed state as the output. Therefore, any measurement can be performed on the transformed state. Let us describe this measurement by a set of POVM (Positive operator valued measurement) elements $\{N_{b}\}$, where $b\in\{1,\cdots,m\}$. The protocol is successful in distinguishing the unitaries if $b$ is the same as $i$. Any classical post-processing of outcome $b$ can be included in the measurement $\{N_{b}\}$. Distinguishability of a set of unitaries consequently reduces to the distinguishability of evolved states. Two unitaries $U_1$ and $U_2$ are distinguishable if $\min|con\{e^{\mathbbm{i}\theta_j}\}|= 0$, where  $\{e^{\mathbbm{i}\theta_j}\}_j$
are the eigenvalues of $U_1^\dagger U_2$, $con\{e^{\mathbbm{i}\theta_j}\}$ denotes the set of complex numbers that can be written as the convex combinations of $\{e^{\mathbbm{i}\theta_j}\}_j$ and $\min|\cdot|$ denotes the minimum norm over all those complex numbers.
More details of the distinguishability of unitaries can be found in \cite{manna25}.
\subsection{Global discrimination of unitary operations (GD)}
In this section, we present general protocol for global discrimination of $m$ $n$-partite product unitaries 
$\{U_i\}_{i=1}^m = \{\otimes_{k=1}^n U_i^{(k)}\}_{i=1}^m$, where local dimensions may vary across parties. 
The party sequentially selects subsets of subsystems to perform discrimination.
First, an arbitrary subset $\mathbbm{S}_1 \subseteq [n]$ with $|\mathbbm{S}_1| = n_1 \leq n$ is chosen. 
A probing state $\rho_{AB}$ is prepared on $\mathcal{H}_A \otimes \mathcal{H}_B$, where 
$\mathcal{H}_A = \bigotimes_{k \in \mathbbm{S}_1} \mathcal{H}_k$. 
The evolved states are
$
\sigma_{AB}^i = \left( \bigotimes_{k \in \mathbbm{S}_1} U_i^{(k)} \otimes \mathbb{I}_B \right)
\rho_{AB}
\left( \bigotimes_{k \in \mathbbm{S}_1} U_i^{(k)} \otimes \mathbb{I}_B \right)^\dagger .
$
A measurement is performed to distinguish $\{\sigma_{AB}^i\}$. 
If perfect discrimination is achieved, the protocol terminates. 
Otherwise, some unitaries are eliminated, and another subset 
$\mathbbm{S}_2 \subseteq [n]\setminus \mathbbm{S}_1$ is selected. 
The procedure is repeated until all subsystems are exhausted.
Each ordered sequence of subsets defines a strategy corresponding to a partition of $[n]$. 
The user considers all permutations of subsets across all partitions and applies any strategy that perfectly distinguishes the unitaries.

\subsection{LOCC discrimination of unitary operations (LD)}
Local discrimination strategy involves $n$ spatially separated parties, indexed by $k = \{1,\cdots,n\}$, such that each party has access to one partite from the $n$-partite unidentified unitary 
$U_i = U^{(1)}_{i}\otimes U^{(2)}_{i}\otimes\cdots\otimes U^{(n)}_{i} = \bigotimes_{k=1}^{n} U^{(k)}_{i}$. 
Such a protocol allows for many possible sequences in which the parties may participate.
Suppose the protocol begins with the party $k$, who prepares a state $\rho_{A_kB_k}$ and applies $U^{(k)}_{i}$ from the unknown unitary $U_{i}$. On the evolved state, the party executes a quantum measurement and obtains an outcome. Based on this outcome, the party rules out certain possibilities for $U^{(k)}_{i}$, and consequently for $U_i$, and communicates the remaining candidates for $U_i$ to next party, say $k'$. 

Using this updated information, the $k'$-th party prepares the best probing state $\rho_{A_{k'}B_{k'}}$, applies the corresponding local unitary, and further eliminates some of the remaining possibilities for $U_i$. This procedure continues sequentially, with each party refining the set of candidate unitaries through local quantum operations and classical communication.

Importantly, the choice of the next party depends on both the classical outcome obtained in the previous step and the information accumulated in earlier rounds. In this way, the $n$ spatially separated parties iteratively reduce the set of possible unitaries until a unique unitary is identified. In this protocol, we take only the one-way communication between the parties.

\section{Local Marking of Unitary Operations (LM)}
Any $r$ number of unitaries are randomly chosen from a known set of $n$-partite $m$ unitaries $\{U_i\}_{i=1}^m = \{\otimes_{k=1}^n U_i^{(k)}\}_{i=1}^m$ and they are distributed among spatially separated $n$ parties without revealing the identity of each unitary. The task of the parties is to perfectly mark each of the unitaries by using local operation and classical communication. We denote this task as $r$-LM. $r$ can take all values from $1$ to $m$, where $1$-LM is the task of LD of a set of unitaries. 

For an example, let us take a set $\mathcal{T}$ consisting of three bipartite unitaries, i.e., $m=3$ and $n=2$. To successfully execute $r$-LM, we need to consider all possible $r$-combinations from these three unitaries, where $r \leq 3$. 

For $2$-LM, the parties need to mark three sets of unitaries, namely $\{U_1,U_2\}$, $\{U_1,U_3\}$, and $\{U_2,U_3\}$. To mark the first set, Alice and Bob need to distinguish the set $\{U_1 \otimes U_2,\, U_2 \otimes U_1\}$ within the LOCC paradigm. Similarly, for marking the other two pairs, they need to distinguish the sets $\{U_i \otimes U_j\}_{i \neq j}$. If they can successfully mark all three sets, then the set $\mathcal{T}$ is said to be $2$-markable ($2$-LM). 

On a similar note, $3$-LM implies the local distinguishability of the set 
$\{U_1 \otimes U_2 \otimes U_3,\, 
U_1 \otimes U_3 \otimes U_2,\, 
U_2 \otimes U_1 \otimes U_3,\, 
U_2 \otimes U_3 \otimes U_1,\, 
U_3 \otimes U_1 \otimes U_2,\, 
U_3 \otimes U_2 \otimes U_1\}$.

\section{Results}
We start with the results which encapsulates the relation between three protocols- GD, LD and LM. Our first theorem provides a sufficient condition of perfect local marking of a set of unitaries.
\begin{thm}\label{1}
    Perfect LD always implies perfect LM.
\end{thm}
\begin{proof}
We start with an arbitrary set of $m$ unitaries and denote the set by $\mathcal{S}\equiv\{U_j\}_{j=1}^m$. We assume this set of unitaries be locally distinguishable. We want to look at the marking properties of this set. This problem can be interpreted as a distinguishability problem for the set of unitaries $\mathcal{S}_{\mathcal{P}[\{m\}]}\equiv\{\mathcal{P}(\otimes_{j=1}^m U_j)\}$, where $\{\mathcal{P}(\otimes_{j=1}^m U_j)\}$ denotes the set of unitaries in the form of tensor product obtained through permutations of the $j$ indices. The unitaries belonging to the set $\mathcal{S}_{\mathcal{P}[\{m\}]}$ can be further decomposed into disjoint groups $Q_l$, where we express these groups as
\begin{align}
    Q_l=U_l\otimes\mathcal{S}_{\mathcal{P}[\{m\}/l]}
\end{align}
We started with the assumption that the set of unitaries $\{U_j\}_j$ are locally distinguishable, then by local operations on the first part of these groups, we can mark the exact group of the particular unitary. Due to the LOCC protocol, the second part does not change. This allows us to further partition the group of unitaries and rewrite them as
\begin{align}
    Q_{l,l'}=U_{l}\otimes U_{l'}\otimes\mathcal{S}_{\mathcal{P}[\{m\}/\{l,l'\}]}
\end{align}
where $l$ denotes the marked unitary after first LOCC has been performed.

By similar arguments, we can mark the index $l'$ through LOCC on the $U_{l'}$ part. Therefore, continuing this protocol, we can show that we can perfectly mark the set of unitaries $\{U_j\}_j$. This completes the proof. 
\end{proof}
At this point, the obvious question arises if the reverse statement of Theorem \ref{1} is true? In the next result, we prove much stronger statement which consequently gives a counter argument of the reverse statement.
\begin{thm}\label{thm2}
 There exist two globally indistinguishable bipartite unitaries $V_1=V_1^A\otimes V_1^B$ and $V_2=V_2^A\otimes V_2^B$ that can be marked successfully. The necessary and sufficient condition is,
 \bea\label{c1}
&(i)& con\{e^{\mathbbm{i}(\theta_i-\theta_j)}\}_{i,j=1}^d=0 \nonumber\\
&&\text{  or  } con\{e^{\mathbbm{i}(\Theta_i-\Theta_j)}\}_{i,j=1}^d=0;\\
&(ii)& con\{e^{\mathbbm{i}(\theta_i+\Theta_j)}\}_{i,j=1}^d\neq 0,\label{c2}
 \eea
 where $eig((V_1^A)^\dagger V_2^A)\in\{e^{\mathbbm{i}\theta_j}\}_j$ and $eig((V_1^B)^\dagger V_2^B)\in\{e^{\mathbbm{i}\Theta_j}\}_j$. $eig(y)$ denotes the eigenvalues of $y$.
\end{thm}
\begin{proof}
Firstly, we impose the condition of indistinguishability on the unitaries $V_1$ and $V_2$. The eigenvalues of $V_1^\dagger V_2$ can be written as
\bea
eig(V_1^\dagger V_2)
&=& eig\!\left((V_1^A)^\dagger V_2^A\right)\cdot 
  eig\!\left((V_1^B)^\dagger V_2^B\right)\nonumber\\
&=& \{e^{\mathbbm{i}(\theta_i+\Theta_j)}\}_{i,j}.
\eea
From subsection~\ref{ssA}, perfect indistinguishability requires that the convex hull of the eigenvalues of $V_1^\dagger V_2$ does not contain the origin. This condition is sufficient to ensure~\eqref{c2}.

Next, successful marking corresponds to successful local distinguishability between the unitaries $V_1\otimes V_2$ and $V_2\otimes V_1$. Since each party has access to two local unitaries, irrespective of the LOCC protocol, the final step necessarily reduces to distinguishing between two unitaries at one party. Hence, a successful LOCC protocol requires perfect distinguishability of the corresponding local unitaries at each party.

For Alice, this amounts to discriminating between $V_1^A\otimes V_2^A$ and $V_2^A\otimes V_1^A$. We compute
\begin{align}
eig\!\left((V_1^A)^\dagger V_2^A \otimes (V_2^A)^\dagger V_1^A\right)
&= eig\!\left((V_1^A)^\dagger V_2^A\right)\cdot 
   eig\!\left((V_2^A)^\dagger V_1^A\right) \nonumber \\
&= \{e^{\mathbbm{i}\theta_i}\}_i \cdot \{e^{-\mathbbm{i}\theta_j}\}_j \nonumber \\
&= \{e^{\mathbbm{i}(\theta_i-\theta_j)}\}_{i,j=1}^d .
\end{align}
Imposing the condition of perfect distinguishability on Alice's side yields the first part of~\eqref{c1}. An analogous calculation on Bob's side gives the second part of~\eqref{c1}.
\end{proof}

An example in this direction makes it more clear.
\begin{example*}
    Let us consider two unitaries acting on $\mathbb{C}^2\otimes\mathbb{C}^2$ as follows:
    \bea
V'_1 &=& \ket{0}\bra{0}+e^{\mathbbm{i}\alpha_1}\ket{1}\bra{1}\otimes\ket{0}\bra{0}+e^{\mathbbm{i}\alpha_2}\ket{1}\bra{1},\nonumber\\
V'_2 &=& \ket{0}\bra{0}+e^{-\mathbbm{i}\alpha_3}\ket{1}\bra{1}\otimes\ket{0}\bra{0}+e^{-\mathbbm{i}\alpha_4}\ket{1}\bra{1},\nonumber\\
    \eea
such that $\alpha_1+\alpha_2+\alpha_3+\alpha_4<\pi$ and $\alpha_1+\alpha_3>\pi/2$ and $\alpha_1,\alpha_2,\alpha_3,\alpha_4\geq 0$.
\end{example*}
For checking the distinguishability of the unitaries $V_1'$ and $V_2'$, we compute 
\bea
eig(V_1'^{\dagger}V_2')&=&eig\Big((V_1'^A)^{\dagger}(V_2'^A)\Big)\cdot eig\Big((V_1'^B)^{\dagger}(V_2'^B)\Big)\nonumber\\
    &=&\left\{1,e^{-\mathbbm{i}(\alpha_1+\alpha_3)}\right\}\cdot \left\{1,e^{-\mathbbm{i}(\alpha_2+\alpha_4)}\right\}\nonumber\\
     &=&\left\{1,e^{-\mathbbm{i}(\alpha_1+\alpha_3)},e^{-\mathbbm{i}(\alpha_2+\alpha_4)},e^{-\mathbbm{i}(\alpha_1+\alpha_2+\alpha_3+\alpha_4)}\right\}.\nonumber\\
\eea
Since $\alpha_1+\alpha_2+\alpha_3+\alpha_4<\pi$, convex combination of above four eigenvalues does not include the origin. Hence, the two unitaries are indistinguishable.

For local marking, suppose Alice starts the protocol. The marking scheme boils down to discriminating between $(V'_1)^A\otimes (V_2')^A$ and $(V'_2)^A\otimes (V_1')^A$. We have
\bea   && eig\Big((V_1'^A)^{\dagger}V_2'^A\otimes(V_2'^A)^{\dagger}V_1'^A\Big)\nonumber\\
&=&eig\Big((V_1'^A)^\dagger V_2'^A\Big)\cdot eig\Big((V_2'^A)^\dagger V_1'^A\Big)\nonumber\\
    &=&\left\{1,e^{-\mathbbm{i}(\alpha_1+\alpha_3)}\right\}\cdot \left\{1,e^{\mathbbm{i}(\alpha_1+\alpha_3)}\right\}\nonumber\\
    &=&\left\{1,e^{-\mathbbm{i}(\alpha_1+\alpha_3)},e^{\mathbbm{i}(\alpha_1+\alpha_3)}\right\}.
\eea
Since $\alpha_1+\alpha_3>\frac{\pi}{2}$, the convex hull of the above set of eigenvalues will contain the origin. If Bob starts the protocol, Bob can not mark the unitaries. But upon trivial communication, Alice can mark the unitaries in the similar way.

Moving on, we ask the opposite question. Can distinguishability is sufficient condition for perfect markability? Next two theorems provide an answer to this.
 
\begin{thm}
    Two bipartite distinguishable unitaries are always markable.
\end{thm}
\begin{proof}
    Let us take two distinguishable bipartite unitaries $W_1=W_1^A\otimes W_1^B$ and $W_2=W_2^A\otimes W_2^B$, where  $eig((W_1^A)^\dagger W_2^A)\in\{e^{\mathbbm{i}\omega_j}\}_j$ and $eig((W_1^B)^\dagger W_2^B)\in\{e^{\mathbbm{i}\Omega_j}\}_j$. The unitaries are distinguishable, which entails,
    \be
    con\{e^{\mathbbm{i}(\omega_i+\Omega_j)}\}_{i,j=1}^d= 0.
    \ee
    The above equation is same as,
    \be\label{d_con}
con\{1,e^{\mathbbm{i}(\omega_k-\omega_1)},e^{\mathbbm{i}(\Omega_k-\Omega_1)},e^{\mathbbm{i}(\omega_k+\Omega_k-\omega_1-\Omega_1)})\}_{k=2}^d= 0.
    \ee

Suppose, $W_1$ and $W_2$ are not markable. The relevant criteria is,
\bea
&(i)&con\{e^{\mathbbm{i}(\omega_i-\omega_j)}\}_{i,j}\neq 0,\nonumber\\
&(ii)&con\{e^{\mathbbm{i}(\Omega_i-\Omega_j)}\}_{i,j}\neq 0.
\eea
The above convex sets contain $1$. Therefore, the above conditions translate into the following facts:
\bea\label{1_2}
&&|\omega_i-\omega_j|<\pi/2, \forall i,j;\nonumber\\
&&|\Omega_i-\Omega_j|<\pi/2, \forall i,j.
\eea
This implies, $|\omega_i-\omega_j|+ |\Omega_i-\Omega_j|<\pi$. This inference establishes the impossibility of \eqref{d_con}, since all points in the corresponding convex set lie entirely within the first two quadrants, and none of them equals 
$-1$. This makes $W_1$ and $W_2$ not distinguishable which emerges contradictory with our initial assumption. On the other way, if \eqref{d_con} holds to be true, at least any of the below equations needs to be true:
\bea
&(i)&|\omega_k-\omega_1|\geq \pi,\nonumber\\
&(ii)&|\Omega_k-\Omega_1|\geq \pi,\nonumber\\
&(iii)&|\omega_k+\Omega_k-\omega_1-\Omega_1|\geq \pi.
\eea
These contradict \eqref{1_2}, which means the unitaries are markable. This completes the proof. 
\end{proof}
This property does not hold if we increase the number of the parties. Let us consider two tripartite unitaries acting on $\mathbb{C}^2\otimes\mathbb{C}^2\otimes\mathbb{C}^2$ as following:
 \bea\label{W_prime}
W'_1 &=& \ket{0}\bra{0}+e^{\mathbbm{i}\beta_1}\ket{1}\bra{1}\otimes\ket{0}\bra{0}+e^{\mathbbm{i}\beta_2}\ket{1}\bra{1}\nonumber\\
&&\otimes\ket{0}\bra{0}+e^{\mathbbm{i}\beta_3}\ket{1}\bra{1},\nonumber\\
W'_2 &=& \ket{0}\bra{0}+e^{-\mathbbm{i}\beta_4}\ket{1}\bra{1}\otimes\ket{0}\bra{0}+e^{-\mathbbm{i}\beta_5}\ket{1}\bra{1},\nonumber\\
&&\otimes\ket{0}\bra{0}+e^{-\mathbbm{i}\beta_6}\ket{1}\bra{1}.
    \eea
\begin{thm}
    Two tripartite unitaries described at \eqref{W_prime} are distinguishable but not markable when $\sum_{i=1}^6 \beta_i=\pi$ and $\beta_1+\beta_4<\pi/2, \beta_2+\beta_5<\pi/2$ and $\beta_3+\beta_6<\pi/2$.
\end{thm}
\begin{proof}
       For $W'_1$ and $W'_2$ to be distinguishable, we need,
       \bea
con\{eig(W^{'\dagger}_1 W'_2)\} &=& 0;
       \eea
  hence, 
  \bea
&&con\left\{1,e^{\mathbbm{i}(\beta_1+\beta_4)},e^{\mathbbm{i}(\beta_2+\beta_5)},e^{\mathbbm{i}(\beta_3+\beta_6)},e^{\mathbbm{i}(\beta_1+\beta_4+\beta_2+\beta_5)},\right.\nonumber\\
&&\left. e^{\mathbbm{i}(\beta_1+\beta_4+\beta_3+\beta_6)}, e^{\mathbbm{i}(\beta_3+\beta_6+\beta_2+\beta_5)}, e^{\mathbbm{i}(\beta_1+\beta_4+\beta_2+\beta_5+\beta_3+\beta_6)}\right\}=0\nonumber\\
  \eea
and this is, indeed, possible as $\sum_i \beta_i=\pi$.

For marking, Alice's criteria is,
\bea
con\{1,e^{\mathbbm{i}(\beta_1+\beta_4)},e^{-\mathbbm{i}(\beta_1+\beta_4)}\} &=& 0.
\eea
As $\beta_1+\beta_4<\pi/2$, the above equation is not valid. Similarly, one can check Bob and Charlie too can not mark these two unitaries as $\beta_2+\beta_5<\pi/2$ and $\beta_3+\beta_6<\pi/2$.
\end{proof}
This result can be read as the stronger version of \emph{nonlocality without entanglement}. From last two theorems, we can infer that, the lowest number parties for this phenomenon is three when only two unitaries are considered. The minimal number of unitaries for bipartite setting is left as a future problem.

All the previous theorems address the relationship between three different tasks—GD, LD, and LM. In all these cases, only two unitary operations are considered. We now turn our attention to more general marking schemes and investigate the (non-)equivalence between these protocols.
\begin{thm}
    $r$-LM always imply $s$-LM, where $s<r$.
\end{thm}
\begin{proof}
    We start with an arbitrary set of $m$ unitaries and denote it by $\mathcal{S}\equiv\{U_j\}_{j=1}^m$. We assume that any $r$ number of unitaries can be marked successfully, that means this set is $r$-LM. The marking properties of a set of $r$ unitaries boils down to a LOCC distinguishability problem for the set of unitaries $\mathcal{S}_{\mathcal{P}[\{r\}]}$. Similarly, this set can further be decomposed into disjoint groups $M_l$ as,
    \begin{equation}\label{decompose}
        M_l=U_l\otimes\mathcal{S}_{\mathcal{P}[\{r\}/l]}=U_l\otimes\mathcal{S}_{\mathcal{P}[\{r-1\}]},
    \end{equation}
where $l$ can take values $\{1,\cdots,r\}$.
    
    For each group $M_l$ corresponding to a fixed $l$, the marking scheme reduces to the LOCC distinguishability of all possible sets obtained by considering permutations of $(r-1)$ unitaries, denoted by $\mathcal{S}_{\mathcal{P}[\{r-1\}]}$, where the unitary $U_l$ is excluded. 
Since $U_l$ is fixed, the discrimination scheme does not extract any information from it. Therefore, in order for the set of $r$ unitaries to be locally distinguishable, the sets $\mathcal{S}_{\mathcal{P}[\{r-1\}]}$ must themselves be locally distinguishable. This implies that $(r-1)$-LM is possible. 

By applying a similar decomposition to $\mathcal{S}_{\mathcal{P}[\{r-1\}]}$, one can further show that $(r-2)$-LM is also possible. Continuing this argument recursively, we conclude that the given set of unitaries is $s$-LM for all $s<r$.
    
\end{proof}
A natural question that arises is whether $r$-LM implies $s$-LM for $s>r$. The answer, however, is not straightforward. In the following theorem, we provide a counterexample to show that $s$-LM does not necessarily imply $(s+1)$-LM.
\begin{thm}
    $r$-LM does not imply $s$-LM, where $s>r$.
\end{thm}
\begin{proof}
    This proof is by construction. We first consider specific three unitaries where $2$-LM is possible. Then we prove that $3$-LM can not be done successfully for the same set of unitaries. Let us take three unitaries as follows:
    \bea\label{u_th6}
    Z_1 &=& Z_1^A\otimes T\nonumber\\
    Z_2 &=& Z_2^A\otimes T\nonumber\\
    Z_3 &=& Z_3^A\otimes T,
    \eea
where $Z_1^A=\ket{0}\bra{0}+\mathbbm{i}\ket{1}\bra{1}, Z_2^A=\ket{0}\bra{0}+\ket{1}\bra{1}$ and $Z_3^A=\ket{0}\bra{0}+e^{\mathbbm{i}\frac{5\pi}{4}}\ket{1}\bra{1}$. First, we prove the perfect $2$-LM for this set of unitary operations. As Bob has the same unitary on his side for the whole set, the local distinguishability reduces to the distinguishability of Alice's unitaries. For $2$-LM, we need to distinguish between three sets of unitaries:
\bea
\mathcal{A}_1 &=& \{Z_1^A\otimes Z_2^A, Z_2^A\otimes Z_1^A \}\nonumber\\
\mathcal{A}_2 &=& \{Z_1^A\otimes Z_3^A, Z_3^A\otimes Z_1^A \}\nonumber\\
\mathcal{A}_3 &=& \{Z_3^A\otimes Z_2^A, Z_2^A\otimes Z_3^A \}
\eea
For the set $\mathcal{A}_1$, we can calculate $eig((Z_2^A)^\dagger Z_1^A\otimes (Z_1^A)^\dagger Z_2^A)= \{1,\mathbbm{i},-\mathbbm{i}\}$. It is easy to check $\min |con\{1,\mathbbm{i},-\mathbbm{i}\}|=0$. In a similar fashion, for the set $\mathcal{A}_2$, $\min |con\{1,e^{\mathbbm{i}\frac{5\pi}{4}},e^{-\mathbbm{i}\frac{5\pi}{4}}\}|=0$ and for the set $\mathcal{A}_3$, $\min |con\{1,e^{\mathbbm{i}\frac{3\pi}{4}},e^{-\mathbbm{i}\frac{3\pi}{4}}\}|=0$. Henceforth, the unitaries are $2$-LM.

Now, we want to check $3$-LM of the unitaries of \eqref{u_th6}. For this proof, we need to check the distinguishability of following six unitaries:
\[
\{Z_1^A\otimes Z_2^A\otimes Z_3^A, Z_1^A\otimes Z_3^A\otimes Z_2^A, Z_2^A\otimes Z_1^A\otimes Z_3^A,
\]
\[
Z_2^A\otimes Z_3^A\otimes Z_1^A,Z_3^A\otimes Z_1^A\otimes Z_2^A,Z_3^A\otimes Z_2^A\otimes Z_1^A\}.
\]
The most general probe for this task would be a $6$-qubit entangled state, which can be written as $\ket{\phi_p}=\sum_{i,j,k,l,m,n=0}^1 C_{i,j,k,l,m,n}\ket{i}\ket{j}\ket{k}\ket{l}\ket{m}\ket{n}$. The unitaries will act on the first three qubits. Let us check the distinguishability of last two unitaries, which implies that,
$\la \psi_p|(Z_3^A)^\dagger Z_3^A\otimes (Z_1^A)^\dagger Z_2^A\otimes (Z_2^A)^\dagger Z_1^A\otimes\I\otimes\I\otimes\I|\psi_p\ra = 0$. Consequently,
\bea\label{con_6}
&&\la \psi_p|(Z_3^A)^\dagger Z_3^A\otimes (Z_1^A)^\dagger Z_2^A\otimes (Z_2^A)^\dagger Z_1^A\otimes\I\otimes\I\otimes\I|\psi_p\ra\nonumber\\
&=&\sum_{i,j,j',k,k',l,m,n=0}^1 C^*_{i,j,k,l,m,n}C_{ij'k'lmn}  \bra{j}(Z_1^A)^\dagger Z_2^A\ket{j'}\nonumber\\
&&\hspace{1 cm}\bra{k}(Z_2^A)^\dagger Z_1^A\ket{k'}\nonumber\\
&=&\sum_{i,j,k,l,m,n=0}^1 |C_{i,j,k,l,m,n}|^2\bra{j}(Z_1^A)^\dagger Z_2^A\ket{j}\bra{k}(Z_2^A)^\dagger Z_1^A\ket{k}\nonumber\\
&=&1(\mathbf{R}_1)+\mathbbm{i}(\mathbf{R}_2)-\mathbbm{i}(\mathbf{R}_3),
\eea
where
\bea
\mathbf{R}_1 &=& \sum_{i,l,m,n}(|C_{i,j=0,k=0,l,m,n}|^2+|C_{i,j=1,k=1,l,m,n}|^2),\nonumber\\
\mathbf{R}_2 &=& \sum_{i,l,m,n}|C_{i,j=1,k=0,l,m,n}|^2,\nonumber\\
\mathbf{R}_3 &=& \sum_{i,l,m,n}|C_{i,j=0,k=1,l,m,n}|^2.
\eea
By normalization, $\mathbf{R}_1+\mathbf{R}_2+\mathbf{R}_3=1$. So the expression at \eqref{con_6} becomes zero, if and only if $\mathbf{R}_2=\mathbf{R}_3=\frac12$.
Now we calculate the distinguishability of first two unitaries, which implies that,
$\la \psi_p|(Z_1^A)^\dagger Z_1^A\otimes (Z_2^A)^\dagger Z_3^A\otimes (Z_3^A)^\dagger Z_2^A\otimes\I\otimes\I\otimes\I|\psi_p\ra = 0$. Consequently,
\bea\label{con_6_1}
&&\la \psi_p|(Z_1^A)^\dagger Z_1^A\otimes (Z_2^A)^\dagger Z_3^A\otimes (Z_3^A)^\dagger Z_2^A\otimes\I\otimes\I\otimes\I|\psi_p\ra\nonumber\\
&=&\sum_{i,j,j',k,k',l,m,n=0}^1 C^*_{i,j,k,l,m,n}C_{ij'k'lmn}  \bra{j}(Z_2^A)^\dagger Z_3^A\ket{j'}\nonumber\\
&&\hspace{1 cm}\bra{k}(Z_3^A)^\dagger Z_2^A\ket{k'}\nonumber\\
&=&\sum_{i,j,k,l,m,n=0}^1 |C_{i,j,k,l,m,n}|^2\bra{j}(Z_2^A)^\dagger Z_3^A\ket{j}\bra{k}(Z_3^A)^\dagger Z_2^A\ket{k}\nonumber\\
&=&1(\mathbf{R}_1)+e^{\mathbbm{i}\frac{5\pi}{4}}(\mathbf{R}_2)+e^{-\mathbbm{i}\frac{5\pi}{4}}(\mathbf{R}_3).
\eea
From the earlier conditions, we get that $\mathbf{R}_2=\mathbf{R}_3=\frac12$. If we put these values of $\mathbf{R}_2$ and $\mathbf{R}_3$ into \eqref{con_6_1}, the expression reduces to non-zero value. That indicates first two unitaries are not distinguishable with $\ket{\phi_p}$, which successfully distinguishes last two unitaries. That proves the impossibility of common probing state to distinguish six unitaries and subsequently, implies the failure of $3$-LM of the unitaries of \eqref{u_th6}.
\end{proof}
This theorem proves to be strikingly opposite to the case of local state marking, where for a set of product states, $r$-LM implies $s$-LM, where $s>r$ (Corollary $2$ of \cite{Sen}). Needless to say, we take the unitaries which are locally implementable, that means they are product unitaries.

There happens to be some interesting phenomena if the parties have the access to only single systems as the probe. Here we show that with this restricted resource, a set unitaries are not distinguishable but they are successfully markable. For that reason, let us take the following three unitaries:
 \bea\label{u_th5}
    W_1 &=& \I\otimes \I\nonumber\\
    W_2 &=& Z\otimes V\nonumber\\
    W_3 &=& X\otimes \I,
    \eea
where $\I=\ket{0}\bra{0}+\ket{1}\bra{1}$, $Z=\ket{0}\bra{0}-\ket{1}\bra{1}$, $X=\ket{0}\bra{1}+\ket{1}\bra{0}$ and $V=\ket{0}\bra{0}+\mathbbm{i}\ket{1}\bra{1}$.
\begin{thm}
    The unitaries of \eqref{u_th5} are distinguishable with entangled system and not distinguishable with single system but can be marked with single system.
\end{thm}
\begin{proof}
For global discrimination, the party can choose two kind of strategies. Firstly, she may not do any partition and try to discriminate three unitaries of dimension $\mathbbm{C}^2\otimes \mathbbm{C}^2$. It can be easily checked that $eig(V^\dagger\I)=\{1,-\mathbbm{i}\}$, which indicates no pair of unitaries of Bob's side are distinguishable. So only useful unitaries are those at Alice's disposal. Therefore, the scheme reduces to the discrimination of three unitaries acting on $\mathbbm{C}^2$. Now she has the access to only single system, these three unitaries are not distinguishable. But if she got the entangled probe $\ket{\phi^+}=\frac{1}{\sqrt{2}}(\ket{00}+\ket{11})$, the unitaries $\I,Z$ and $X$ are distinguishable as they transform $\ket{\phi^+}$ to three orthogonal bell states.

The party may adopt an alternative strategy in which she first operates on one bipartition and, depending on the measurement outcome, proceeds to discriminate the entire ensemble. Suppose she begins with Alice’s unitaries. Regardless of the choice of a single-system probe, it is impossible to perfectly distinguish the three unitaries. Consequently, there will be at least one instance in which Alice can eliminate at most one unitary from consideration. This leaves Bob with the task of distinguishing between the remaining two unitaries, which are not perfectly distinguishable.

Similarly, the party may choose to operate on Bob’s unitaries first. Since, on Alice’s side, any pair of unitaries is perfectly distinguishable, the requirement shifts to Bob’s side, where the party must eliminate at least one unitary for each measurement outcome. Suppose the party performs a measurement $G$ with POVM elements $\{G_i\}_i$. Consider a two-outcome measurement with elements $G_1$ and $G_2$, where $G_1$ eliminates the possibility of $\mathbbm{I}\ket{\bar{\tau}}$ and $G_2$ eliminates $V\ket{\bar{\tau}}$, with $\ket{\bar{\tau}}$ being the optimal single-system probe for the task. Clearly, a two-outcome measurement suffices for this strategy. However, for this protocol to be valid, $G_1$ and $G_2$ must form legitimate POVM elements. If such elements exist, it would imply that the unitaries $\mathbbm{I}$ and $V$ are perfectly distinguishable, which is a contradiction. Henceforth, the unitaries at \eqref{u_th5} are not distinguishable with single system probe.

For marking Alice and Bob needs to discriminate the following set of unitaries via LOCC:
\bea
W_{123}&=&(\I\otimes Z\otimes X)_A\otimes(\I\otimes V\otimes\I)_B\nonumber\\
W_{132}&=&(\I\otimes X\otimes Z)_A\otimes(\I\otimes \I\otimes V)_B\nonumber\\
W_{213}&=&(Z\otimes \I\otimes X)_A\otimes(V\otimes\I\otimes\I)_B\nonumber\\
W_{231}&=&(Z\otimes X\otimes I)_A\otimes(V\otimes\I\otimes\I)_B\nonumber\\
W_{312}&=&(X\otimes \I\otimes Z)_A\otimes(\I\otimes \I\otimes V)_B\nonumber\\
W_{321}&=&(X\otimes Z\otimes \I)_A\otimes(\I\otimes V\otimes\I)_B
\eea
The suffix $A$ and $B$ denotes the unitaries at Alice's and Bob's disposal respectively. Suppose Alice starts the protocol by using the probe $\ket{0}\ket{+}\ket{0}$. With this probe, the set of transformed states would be:
\[
\big\{\ket{\eta_1}=\ket{0}\ket{+}\ket{1},\ket{\eta_2}=\ket{0}\ket{+}\ket{0},\ket{\eta_3}=\ket{0}\ket{-}\ket{1},
\]
\[\ket{\eta_4}=\ket{0}\ket{+}\ket{0},\ket{\eta_5}=\ket{1}\ket{+}\ket{0},\ket{\eta_6}=\ket{1}\ket{-}\ket{0}\big\}
\]
Note that, $\ket{\eta_2}=\ket{\eta_4}$.
Alice will choose a measurement whose POVM elements are
\[
\Big\{\ket{\eta_1}\bra{\eta_1},\ket{\eta_2}\bra{\eta_2},\ket{\eta_3}\bra{\eta_3},\ket{\eta_5}\bra{\eta_5},\ket{\eta_6}\bra{\eta_6},
\]
\[
\I-(\sum_{i=1,2,3,5,6}\ket{\eta_i}\bra{\eta_i})\Big\}
\]
Needless to say, the last outcome does not occur in this protocol. For each outcome, Alice can perfectly identify one unitary, except for the outcome corresponding to $\ket{\eta_2}\bra{\eta_2}$. When this outcome clicks, Alice knows that the unitary is either $W_{132}$ or $W_{231}$, and she communicates this information to Bob. 
Bob then has to distinguish between $\I\otimes\I\otimes V$ and $V\otimes\I\otimes\I$. We can check $eig(\I V\otimes\I\otimes V^\dagger \I)=\{1,\mathbbm{i},-\mathbbm{i}\}$. A suitable convex combination of these eigenvalues can yield zero. Therefore, these unitaries are distinguishable with a single system probe. 
Thus, these six unitaries are distinguishable via LOCC. This implies that the unitaries at \eqref{u_th5} are markable with a single-system probe.
\end{proof}
\section{Conclusion}
In this work, we introduce a novel discrimination task for locally implementable unitary operations, which is termed as \emph{local marking}. We show that this task is fundamentally different with respect to unitary discrimination as well as, this task is not a trivial generalization of marking of quantum states. We prove that perfect local discrimination implies perfect local marking; however, local marking does not imply either LOCC distinguishability or even global distinguishability. We derive a necessary and sufficient condition for the successful local marking of two indistinguishable unitaries. Furthermore, we prove that any two bipartite distinguishable unitaries are always markable.
A striking feature emerges in the multipartite setting: we demonstrate the existence of two tripartite globally distinguishable unitaries that cannot be locally marked, revealing a phenomenon analogous to \emph{nonlocality without entanglement}. We also show that while marking a set of unitaries implies the marking of its subsets, the converse does not hold. Finally, we present examples where entangled probes are required for discrimination, whereas local marking of the same set can be achieved using product probes.

Our work leaves several open directions for future research. The most immediate one is the extension of local marking to general quantum channels. While most of our results focus on bipartite systems, exploring the multipartite ($n$-partite) setting is a natural and promising direction.
There has been significant progress in distinguishability assisted by shared correlations among parties, which enable otherwise impossible local tasks; investigating similar resources in the context of local marking would be highly valuable. Finally, local marking may find applications in quantum information–theoretic tasks such as secret sharing and cryptographic protocols.
\bibliography{ref}
\end{document}